\documentclass[11pt]{article}

\pdfoutput=1

\usepackage[margin=1in]{geometry}

\usepackage[T1]{fontenc}

\usepackage{listings}
\usepackage{ifthen}

\usepackage{enumitem}
\setlist{itemsep=0pt,parsep=0pt}             % more compact lists

\usepackage{booktabs}
\usepackage[usenames,dvipsnames,svgnames,table]{xcolor}
\usepackage{colortbl}
\usepackage{multirow}

\usepackage[ruled,vlined,linesnumbered,noend]{algorithm2e}
\SetNlSty{textsf}{}{:}
\DontPrintSemicolon
\SetKwComment{tcp}{$\rhd$\, }{}

\usepackage{amsmath,amssymb}
\usepackage{ktmath}
\usepackage{mathtools}
\usepackage{mathabx}
\usepackage{xfrac}

\usepackage{txfonts}                            % otherwise, Adobe Times and Math Font
\usepackage[scaled=0.9]{helvet}
\usepackage{graphicx}
\usepackage{wrapfig}
\usepackage{balance}  % for  \balance command ON LAST PAGE  (only there!)
\usepackage{tikz}
\usepackage{alltt}
\usetikzlibrary{positioning}

\newcommand{\note}[2]{
    \ifthenelse{\equal{\showComments}{yes}}{\textcolor{#1}{#2}}{}
}

\date{}
\title{Parallel Streaming Random Sampling}
\author{
  Kanat Tangwongsan\thanks{
    CS Program, Mahidol University International College,~\texttt{kanat.tan@mahidol.edu}}
  \and
  Srikanta Tirthapura\thanks{
    Dept. of Electrical and Computer Engineering, Iowa State University,~\texttt{snt@iastate.edu}
  }
}

\DeclareMathOperator*{\argmax}{arg\,max}
\newcommand{\qed}[0]{}
\newcommand{\ltheading}[1]{\subsection{#1}}

\newcommand*{\techniqueName}{$\textsf{R}^3$\xspace}
\newcommand{\whp}{\textbf{whp}\xspace}

\newcommand{\stream}{\mathcal{S}}
\newcommand{\skipp}{\textsc{Skip}}
\newcommand{\sworinf}{\textsc{SWOR-Infwin}}
\newcommand{\sworfixed}{\textsc{SWOR-Fixed-Win}}

\newcommand{\swrinf}{\textsc{SWR-Infwin}}

\newcommand{\hyper}{\mathcal{H}}
\newcommand{\binomial}{\mathcal{B}}
\newcommand{\remove}[1]{}

\newcommand{\oldbkt}{\textsc{oldbkt}}
\newcommand{\newbkt}{\textsc{newbkt}}
\newcommand{\oldsmp}{\textsc{oldsmp}}
\newcommand{\newsmp}{\textsc{newsmp}}

\newcommand*{\Stream}{\mathcal{S}}
\newcommand*{\coinHeads}[0]{\textsf{\textit{H}}}
\newcommand*{\uniformSample}[0]{\lstinline{UniformSample}}
\newcommand*{\coinFlip}[0]{\upshape\lstinline{coin}}
\newcommand*{\topsAlg}{\normalfont\textsc{Priority Sampling}\xspace}
\newcommand*{\simpleRRAlg}[0]{\normalfont\textsc{Simple-RR}}
\newcommand*{\fastSingleRRAlg}[0]{{\normalfont\textsc{Fast-Single-RR}}}
\newcommand*{\fastRRAlg}[0]{{\normalfont\textsc{Fast-RR}}}
\newcommand*{\subsegX}{X_{\bigdot}}
\newcommand*{\createViewAlg}{\textsc{Create-View}\xspace}
\newcommand{\ssequence}[1]{\ensuremath{\langle #1 \rangle}}

\usepackage{microtype}
\begin{document}

\maketitle

\begin{abstract}
  This paper investigates parallel random sampling from a potentially-unending
  data stream whose elements are revealed in a series of element sequences
  (minibatches). While sampling from a stream was extensively studied
  sequentially, not much has been explored in the parallel context, with prior
  parallel random-sampling algorithms focusing on the static batch model.
  We present parallel algorithms for minibatch-stream sampling in two settings:
  (1) sliding window, which draws samples from a prespecified number of
  most-recently observed elements, and (2) infinite window, which draws samples
  from all the elements received. Our algorithms are computationally and memory
  efficient: their work matches the fastest sequential counterpart, their
  parallel depth is small (polylogarithmic), and their memory usage matches the
  best known.
 
% We consider random sampling from a potentially-unending data stream whose elements are revealed in a series of element sequences, known as minibatches. Random sampling from a stream is a fundamental problem which is well-studied in the sequential setting. While prior work on parallel random sampling has considered the case when all data elements are available at once (the batch model), but not much is known for parallel random sampling when elements arrive in a streaming fashion. We present efficient parallel algorithms for random sampling in the infinite-window setting, where a sample is drawn from all the elements seen so far, and the sliding window setting, where a sample is desired from the $w$ most recently observed elements. Our parallel algorithms for processing a minibatch are computationally efficient---their work matches efficient sequential algorithms, and their parallel depth is small (polylogarithmic). They are also efficient memory-wise; their memory usage matches the most space-efficient streaming algorithms known.
\end{abstract}

%In all settings considered, the goal is to deliver $s$ uniform samples when queried. In the infinite-window setting, a random sample is drawn from all the data elements seen so far. We prove a lower bound on work and derive a work-optimal  algorithm. The algorithm also uses optimal space. In the sliding-window
%  setting, the sample is desired from the latest $w$ elements. We present lower
%  bounds in terms of work and algorithms for both the settings where $w$ is
%  fixed in advance, and where an upper bound $W$ is given but any $w \leq W$ can
%  be queried. Our algorithms are optimal in terms of space and almost match the
%  lower bounds in terms of work. All algorithms in the paper have
%  polylogarithmic depth. Additionally, we show that sampling $s$ elements with
%  replacement in the infinite window setting can be done more efficiently than
%  the na\"ive method of maintaining $s$ independent runs of single-element
% samplers.
  % Our algorithms yield an efficient parallelization of stream
  % sampling, for which sequential algorithms have been heavily studied in the
  % past.

%%% Local Variables:
%%% mode: latex
%%% TeX-master: "paper"
%%% End:

\section{Introduction}
%We consider high-throughput parallel algorithms for sampling from a massive data stream. While a majority of the work in on stream sampling has focused streaming algorithms has focused on the {\em space complexity} of performing a certain aggregate task, such as sampling, our focus is on improving the processing throughput  using parallel computing, while maintaining the desirable properties of low space complexity.
Consider a model of data processing where data is revealed to the processor in a
series of element sequences (minibatches) of varying sizes. A minibatch must be
processed soon after it arrives. However, the data is too large for all the
minibatches to be stored within memory, though the current minibatch is
available in memory until it is processed.

Such a minibatch streaming model is a generalization of the traditional data
stream model, where data arrives as a sequence of elements. If each minibatch is
of size $1$, our model reduces to the streaming model. Use of minibatches are
common. For instance, in a \emph{data stream warehousing system}~\cite{JS15},
data is collected for a specified period (such as an hour) into a minibatch and
then ingested, and statistics and properties need to be maintained during the
ingestion. Minibatches may be relatively large, potentially of the order of
Gigabytes or more, and could leverage parallelism (e.g., a distributed memory
cluster or a shared-memory multicore machine) to achieve the desired throughput.
Furthermore, this model matches the needs of modern ``big data'' stream
processing systems such as Apache Spark Streaming~\cite{Zaharia13}, where
newly-arrived data is stored as a distributed data set (an ``RDD'' in Spark)
that is processed in parallel. Queries are posed on all the data received up to
the most recent minibatch.

%The advantage of having such large
%minibatches is that it is possible to process this data using parallel
%computing. Using the power of multiple processors, such as a multicore machine
%or cluster, it is possible to significantly increase the processing throughput.

This paper investigates the foundational aggregation task of random sampling in
the minibatch streaming model. Algorithms in this model observe a (possibly
infinite) sequence of minibatches $B_1,B_2,\ldots, B_t, \ldots$. We consider the
following variants of random sampling, all of which have been well studied in
the context of sequential streaming algorithms. In the \textbf{infinite window}
model, a random sample is chosen from all the minibatches seen so far. Thus,
after observing $B_t$, a random sample is drawn from $\cup_{i=1}^t B_i$. In the
\textbf{sliding window} model with window size $w$, the sample after observing
$B_t$ is chosen from the $w$ most-recent elements. Typically, the window size
$w$ is much larger than a minibatch size. \footnote{One could also consider a
  window to be the $w$ most recent minibatches, and similar techniques are
  expected to work.} In this work, the window size $w$ is provided at query
time, but an upper bound $W$ on $w$ is known beforehand.
%We consider two variants of the sliding window model. In the {\bf fixed length
 % sliding window} model, the window size $w$ is known in advance---this is
%appropriate when the aggregation function and its parameters are known
%beforehand, as is common in a data streaming system. 
%In the {\bf variable length
%  sliding window model}, the window size $w$ is provided at query time, and an
%upper bound $W$ on $w$ is known beforehand. This is the most general model.
%Fixed length sliding window is obviously a special case. %, and infinite window is also a special case that can be achieved by setting $W=\infty$.

%We are especially interested in the case where the size of a minibatch is large (order of GB or more), so that it is desirable to have time-efficient algorithms for processing a minibatch. 
%Ideally, the time for processing a minibatch must be sub-linear in the size of the minibatch. 

We focus on optimizing the work and parallel depth of our parallel algorithms.
This is a significant point of departure from the traditional streaming
algorithms literature, which has mostly focused on optimizing the memory
consumed. In addition, we consider memory to be a scarce resource and design for
scenarios where the size of the stream is very large---and the stream, or even a
sliding window of the stream, does not fit in memory. Hence, this work strives
for parallel computational efficiency in addition to memory efficiency.

% %----------------------------------------
% \subsection{Contributions}
% \label{sec:contributions}
% %----------------------------------------

%-----------------------------
\ltheading{Our Contributions}
%-----------------------------
%
We present parallel random-sampling algorithms for the minibatch streaming model, in both
infinite window and sliding window settings. These algorithms can use the power
of shared-memory parallelism to speedup the processing of a new minibatch as
well as a query for random samples.

\smallskip
\noindent\emph{$\rhd$~Efficient Parallel Algorithms.}
Our algorithms are provably efficient in parallel processing. We analyze them in
the work-depth model, showing (1)~they are work-efficient, i.e., total work
across all processors is of the same order as an efficient sequential algorithm,
and (2)~their parallel depth is logarithmic in the target sample size, which
implies that they can use processors nearly linear in the input size while not
substantially increasing the total work performed. In the infinite-window case,
the algorithm is work-optimal since the total work across all processors matches
a lower bound on work, which we prove in this paper, up to constant factors.
Interestingly, for all our algorithms, the work of the parallel algorithm is
sublinear in the size of the minibatch.

\smallskip

\noindent\emph{$\rhd$~Small Memory.} While the emphasis of this work is on
improving processing time and throughput, our algorithms retain the property of
having a small memory footprint, matching the best sequential algorithms from
prior work.

%{\color{red} Write about technical contributions, including the difficulty of parallelizing sliding window.}
%{\color{Plum}
%
Designing such parallel algorithms requires overcoming several challenges.
Sliding-window sampling is typically implemented with
\topsAlg{}~\cite{BDM02,BOZ09}, whose work performed (per minibatch) is linear in
the size of the minibatch. Parallelizing it reduces depth but does not reduce
work.
Generating skip offsets, \`{a} la Algorithm Z~\cite{Vitt85} (reservoir
sampling), can significantly reduce work but offers no parallelism. Prior algorithms,
such as in~\cite{Vitt85}, seem inherently sequential, since the next location
to sample from is derived as a function of the previously chosen location.
This work introduces a new technique called \techniqueName{} sampling, which
combines \textbf{\underline{r}}eversed \textbf{\underline{r}}eservoir sampling
with \textbf{\underline{r}}ejection sampling. \techniqueName{} sampling is a new
perspective on \topsAlg{} that mimics the sampling distribution of \topsAlg{},
but is simpler and has less computational dependency, making it amendable to
parallelization. To enable parallelism, we draw samples simultaneously from
different areas of the stream using a close approximation of the distribution.
This leads to slight oversampling, which is later corrected by rejection
sampling. We show that all these steps can be implemented in parallel.
In addition, we develop a data layout that permits convenient update and fast
queries. As far as we know, this is the first efficient parallelization of the
popular reservoir-sampling-style algorithms.
\ltheading{Related Work}
%------------------------
Reservoir sampling (attributed to
Waterman) was known since the 1960s. There has been much follow-up work,
including methods for speeding up reservoir sampling by ``skipping past
elements''~\cite{Vitt85}, weighted reservoir sampling~\cite{ES06}, and sampling
over a sliding window~\cite{BOZ09,XTB08,BDM02,GemullaL08}.

The difference between the distributed streams
model~\cite{Cormode13,cmyz12,GT01,CTW16} considered earlier, and the parallel
stream model considered here is that in the distributed streams model, the focus
is on minimizing the communication between processors while in our model,
processors can coordinate using shared memory, and the focus is on
work-efficiency of the parallel algorithm. Prior work on shared-memory parallel streaming has
considered frequency counting~\cite{DasAAA09a,TTW14} and aggregates on graph
streams~\cite{TPT13}, but to our knowledge, there is none so far on random
sampling. Prior work on warehousing of sample data~\cite{BrownH06} has considered methods for sampling under
minibatch arrival, where disjoint partitions of new data are handled in parallel. Our work also considers
how to sample from a single partition in parallel, and can be used in conjunction with a method such as~\cite{BrownH06}.

%%% Local Variables:
%%% mode: latex
%%% TeX-master: "paper"
%%% End:

%----------------------------------------------
\section{Preliminaries and Notation}
\label{sec:prelim}
%----------------------------------------------
A \emph{stream} $\Stream$ is a potentially infinite sequence of minibatches $B_1
B_2,\ldots$, where each minibatch consists of one or more elements. Let
$\stream_t$ denote the stream so far until time $t$, consisting of all elements
in minibatches $B_1,B_2,\ldots,B_{t}$. Let $n_i = |B_i|$ and $N_t = \sum_{i=1}^t
n_i$, so $N_t$ is the size of $\stream_t$. The size of a minibatch
is not known until the minibatch is received, and the minibatch is received as
an array in memory. A \emph{stream segment} is a finite sequence of consecutive
elements of a stream. For example, a minibatch is a stream segment. A
\emph{window} of size $w$ is the stream segment consisting of the $w$ most
recent elements.

Let $[n]$ denote the set $\{1, \dots, n\}$. For sequence $X = \ssequence{x_1,
  x_2, \dots, x_{|X|}}$, the $i$-th element is denoted by $X_i$ or $X[i]$. For
convenience, negative index $-i$, written $X[-i]$ or $X_{-i}$, refers to the
$i$-th index from the right end---i.e., $X[|X| - i + 1]$. Following common array
slicing notation, let $X[a\texttt{:}]$ be the subsequence of $X$ starting
from index $a$ onward.
An event happens with high probability (\whp) if it happens with
probability at least $1 - n^{-c}$ for some constant $c \geq 1$.
%Let $H_i = \sum_{j=1}^i \frac{1}{j}$ denote the $i$-th Harmonic number. 
%It is well known that $H_k = O(\log k)$. 
Let \uniformSample$(a, b)$, $a \leq b$, be a function that returns an element
from $\{a, a+1, \dots, b\}$ chosen uniformly at random. For $0 < p \le 1$,
\coinFlip($p$)$\;\in \{\textsf{\itshape H}, \textsf{\itshape T}\}$ returns heads
(\textsf{\itshape H}) with probability $p$ and tails (\textsf{\itshape T}) with
probability $1-p$. For $m \le n$, an $m$-permutation of a set $S$, $|S| = n$, is
an ordering of $m$ elements chosen from $S$.

We analyze algorithms in the work-depth model assuming concurrent reads and
arbitrary-winner concurrent writes. The \emph{work} of an algorithm is the total
operation count, and \emph{depth} (also called parallel time or span) is the
length of the longest chain of dependencies within that algorithm. The gold
standard in this model is for an algorithm to perform the same amount of work as
the best sequential counterpart (work-efficient) and to have polylogarithmic
depth.
%Recent work~\cite{DBS18} has shown that
%algorithms designed in this model can have good practical performance. 
This setting has been fertile ground for research and experimentation on
parallel algorithms. Moreover, results in this model are readily portable to
other related models, e.g., exclusive read and exclusive write, with a modest
increase in cost (see, e.g.,~\cite{BM04:book}).

\begin{fulltext}
\emph{Parallel semisorting} is the problem of reordering an input sequence
of keys so that like sorting, equal keys are arranged contiguously, but unlike
sorting, different keys are not necessarily in sorted order. We rely on the
following result:

\begin{theorem}[\cite{GuShunSunBlelloch15}]
\label{thm:semisorting}
On input a sequence $X = \ssequence{x_1, \dots, x_n}$, where $x_i$ can be
uniformly hashed to $[n^k]$ in constant time, parallel semisorting can be
implemented in $O(n)$ expected work and space, and $O(\log n)$ depth \whp.
\end{theorem}

A \emph{Random permutation} of a finite set $n$-element set $X$ can be generated
in parallel using $O(n)$ work and $O(\log n)$ depth (e.g.,~\cite{Reif1985}).
Later, we use the following bound:

\begin{theorem}[Theorem 1.1 of~\cite{DubhashiPanconesi2009}]
 \label{thm:chernoff}
 Let $X = \sum_{i=1}^n X_i$, where $X_i$s are independent random variables in $[0, 1]$.  For $t > 2e\cdot \expct{X}$, we have $\prob{X \geq t} \leq 2^{-t}$.
\end{theorem}
\end{fulltext}

We measure the space complexity of our algorithms in terms of the number of elements stored. Our space bounds do not represent bit complexity. Often, the space used by the algorithm is a random variable, so we present bounds on the expected space complexity.

% Discuss the case when a sample of size less than $s$ is needed.
%%% Local Variables:
%%% mode: latex
%%% TeX-master: "paper"
%%% End:

%\clearpage
%\section{Variable-Length Sliding Window}

\section{Parallel Sampling from a Sliding Window}

\newcommand*{\sworvar}[0]{{\normalfont\textsc{SWOR-Sliwin}}}
\newcommand*{\rsStore}[0]{\mathcal{R}}
\newcommand*{\queryAlg}[0]{{\upshape\textsc{Sample}}}
\newcommand*{\makeAlg}[0]{{\upshape\textsc{Construct}}}

This section discusses parallel algorithms for sampling without replacement from
a sliding window (\sworvar). The task is as follows: For target sample size $s$
and maximum window size $W$, \sworvar{} is to maintain a data structure
$\rsStore$ that supports two operations: (i)~\lstinline{insert}$(B_i)$
incorporates a minibatch $B_i$ of new elements that arrived at time $i$ into
$\rsStore{}$, and (ii)~For parameters $q \leq s$ and $w \le W$,
\lstinline{sample}$(q, w)$ when posed at time $i$ returns a random sample of $q$
elements chosen uniformly without replacement from the $w$ most recent elements
in $\stream_i$.

In our implementation, \lstinline{sample}$(q,w)$ does something stronger, and
returns a $q$-permutation (not only a set) chosen uniformly at random from the
$w$ newest elements from $\rsStore$---this can additionally be used to generate
a sample of any size $j$ from $1$ till $q$ by only consider the first $j$
elements of the permutation.

One popular approach to sampling from a sliding window in the sequential
setting~\cite{BDM02,BOZ09} is the \topsAlg{} algorithm: Assign a random priority
to each stream element, and in response to \lstinline{sample(s, w)} return the
$s$ elements with the smallest priorities among the latest $w$ arrivals. To
reduce the space consumption to be sub-linear in the window size, the idea is to
store only those elements that can potentially be included in the set of $s$
smallest priorities for any window size $w$. A stream element $e$ can be
discarded if there are $s$ or more elements with a smaller priority than $e$
that are more recent than $e$. Doing so systematically leads to an expected
space bound of $O(s + s\log(W/s))$~\cite{BDM02}\footnote{The original algorithm
  stores the largest priorities but is equivalent to our view.}.

As stated, this approach expends work linear in the stream length to
examine/assign priorities, but ends up choosing only a small fraction of the
elements examined. This motivates the question: \emph{How can one determine
  which elements to choose, ideally in parallel, without expending linear work
  to generate or look at random priorities?} We are most interested in the case
where $W \gg n_i \geq s$, where $n_i$ is the size of minibatch $i$. The main
result of this section is as follows:
\begin{theorem}
  \label{thm:main-var-win}
  There is a data structure for {\sworvar{}} that uses $O(s + s\log(W/s))$
  expected space and supports the following operations:
  \begin{enumerate}[topsep=1pt,label=(\roman*)]
  \item \lstinline{insert}$(B)$ for a new minibatch $B$ uses $O(s + s\log(\tfrac{W}{s}))$ work and $O(\log W)$
    parallel depth; and
  \item \lstinline{sample(q, w)} for sample size $q \le s$ and window size $w \le W$ uses $O(q)$ work and $O(\log W)$ parallel depth.
  \end{enumerate} 
\end{theorem}

Note that the work of the data structure for inserting a new minibatch is only
logarithmic in the maximum window size $W$ and independent of the size of the
minibatch. To prove this theorem, we introduce \techniqueName{} sampling, which
brings together reversed reservoir sampling and rejection sampling. We begin by
describing reversed reservoir sampling, a new perspective on priority sampling
that offers more parallelism opportunities. After that, we show how to implement
this sampling process efficiently in parallel via rejection sampling.

\subsection{Simple Reversed Reservoir Algorithm}
\label{sec:simpleRRAlg}
We now describe \emph{reversed reservoir} (RR) sampling, which mimics the
behaviors of priority sampling but provides more independence and more
parallelism opportunities. This process will be refined and expanded in
subsequent sections. After observing sequence $X$, \simpleRRAlg
(Algorithm~\ref{algo:simpleRR}) yields uniform sampling without replacement of
up to $s$ elements for any suffix of $X$. For maximum sample size $s > 0$ and
integer $i > 0$, define \( p_{-i}^{(s)} = \min\left(1, \frac{s}{i} \right) \).

We say the $i$-th most-recent element has age $i$; this position/element will be
called age $i$ when the context is clear. The algorithm examines $X$ in reverse,
$X[-1], X[-2], \ldots$, and stores a subset in data structure $A$, which records
the index of an element in $X$ as well as a slot (from $[s]$) into which the
element is mapped. Multiple elements may be mapped to the same slot. The
probability of $X[-i]$ being chosen into $A$ decreases as $i$ increases.

%, so the sampling
%process can be described as follows:

\begin{algorithm}[h]
  \caption{\simpleRRAlg$(X, s)$ --- Na\"ive reversed reservoir sampling}
  \label{algo:simpleRR} 
  \KwIn{a stream segment $X = \langle x_1, \dots, x_{|X|} \rangle$ and a
    parameter $s > 0$, $s \leq |X|$.}

  \KwOut{a set $\{(k_i, \ell_i)\}$, where $k_i$ is an index into $X$ and $\ell_i
    \in [s]$}

  $\pi \gets $ Random permutation of $[s]$, $A_0 = \emptyset$\\
  \lFor{i = $1, 2, \dots, s$}{$A_i = A_{i-1} \cup \{(i, \pi_i)\}$}
  \For{i = $s+1, 2, \dots, |X|$}{
    \If{\coinFlip$(p_{-i}^{(s)})$ == \coinHeads}{
      $\ell \gets $\uniformSample$(1, s)$\\
      $A_{i} = A_{i-1} \cup \{(i, \ell)\}$
%      $A \gets A \cup \{(i, \ell)\}$
    }
    \lElse {
       $A_{i} = A_{i-1}$
    }
    %$i$ += 1  \label{lno:checkpoint}
  }

  \Return {$A_{|X|}$}
\end{algorithm}

This algorithm samples an element at index $-i$ (age $i$) with probability
$p_{-i}^{(s)} = \min(1, s/i)$, the same probability it would in priority
sampling (aka.bottom-$k$ sampling). Also, every element sampled is assigned a
random \emph{slot} number between $1$ and $s$. This is used to generate a random
permutation.

Reserved reservoir sampling has a number of nice properties:

\begin{lemma}
  For input a stream segment $X$ and a parameter $0 < s \leq |X|$, the number of
  elements sampled by \simpleRRAlg{} is expected $s + O(s\log (|X|/s))$.
  % If $s
  % \in O(|X|^{1 - \vareps})$ for some $\eps > 0$, this happens \whp.
\end{lemma}
\begin{proof}
  The elements at indicies $-1, -2, \dots, -s$ are always chosen, contributing
  $s$ elements to the output. For $i = s+1, \dots, |X|$, the probability that
  $x_{-i}$ is sampled is $s/i$, so the expected number of samples among these
  elements is
  \[
    \sum_{i=s+1}^{|X|} \frac{s}{i} \leq s\int_{x=s}^{|X|} \frac{s}{x} dx =
    s\ln(|X|/s),
  \]
  which completes the proof.
\end{proof}

\newcommand*{\drawPerm}[0]{\chi}
\newcommand*{\lrElt}[0]{\nu}

Let $A$ denote the result of \simpleRRAlg. Using this, sampling $s$ elements
without replacement from any suffix of $X$ is pretty straightforward.  Define
\[
  \drawPerm(A) = (\lrElt_A(1), \lrElt_A(2), \dots, \lrElt_A(s)) %\label{eq:findperm}
\]
where \( \lrElt_A(\ell) = \argmax_{k \geq 1} \{ (k, \ell) \in A \}\)
is\footnote{Because $|X| \geq s$, the function $\lrElt$ is always defined.} the
oldest element assigned to slot $\ell$.
Given $A$, we can derive $A_{i}$ for any $i \le |X|$ by considering the appropriate subset of $A$.
%Since the probability that $x_{-i}$ is sampled into slot $\ell$ is $\prob{(-i,
%  \ell) \in A} = p_{-i}^{(s)}/s$, 
We have that $\drawPerm(A_i)$ is an $s$-permutation of the $i$ most recent elements of $X$. 
\begin{lemma}
  If $R$ is any $s$-permutation of $X[-i\texttt{:}]$, then
  \begin{equation}
    \prob{R = \drawPerm(A_i)} = \frac{(i - s)!}{i!} \label{eq:perm-prob}
  \end{equation}
\end{lemma}

\begin{proof}
  We proceed by induction on $i$. The base case of $i = s$ is easy to verify
  since $\pi$ is a random permutation of $[s]$ and $\drawPerm(A_s)$ is a
  permutation of $X[-s:]$ according to $\pi$. For the inductive step, assume
  that equation \eqref{eq:perm-prob} holds for for any $R$ that is an
  $s$-permutation of $X[-i:]$. Now let $R'$ be an $s$-permutation of
  $X[-(i+1):]$. Consider two cases:
  \begin{itemize}
  \item Case I: $x_{-(i+1)}$ appears in $R'$, say at at $R'_\ell$. For $R' =
    \drawPerm(A_{i+1})$, it must be the case that $x_{-(i+1)}$ was chosen and
    was assigned to slot $\ell$. Furthermore, $\drawPerm(A_i)$ must be identical
    to $R'$ except in position $\ell$, where it could have been any of the $i -
    (s - 1)$ choices.  This happens with probability
    \begin{align*}
        \textstyle
        (i - [s-1])\cdot \frac{(i - s)!}{i!} \cdot p_{-(i+1)}\cdot \frac{1}{s}
      &= \frac{(i-s)!}{i!}\cdot \frac{i - s + 1}{s}\cdot\frac{s}{i+1}\\
      &= \frac{([i+1]-s)!}{(i+1)!}
    \end{align*}

  \item Case II: $x_{-(i+1)}$ does not appear in $R'$. Therefore, $R'$ must be
    an $s$-permutation of $X[-i:]$ and $x_{-(i+1)}$ was not sampled. This
    happens with probability
    \[
      \textstyle\frac{(i-s)!}{i!} \cdot (1 - p_{-(i+1)}) = \frac{(i-s)!}{i!}\cdot
      \frac{i-s+1}{i+1} = \frac{([i+1] -s)!}{(i+1)!}
    \]
  \end{itemize}
  In either case, this gives the desired probability.
\end{proof}

Note that the space taken by this algorithm (the size of $A_{|X|}$) is $O(s + s \log(|X|/s))$, which is
optimal~\cite{GemullaL08}. The steps are easily parallelizable but still need
$O(|X|)$ work, which can be much larger than the $(s+s\log(|X|/s))$ bound on the 
number of elements the algorithm must sample. We improve on this next.

\subsection{Improved Single-Element Sampler}

This section addresses the special case of $s = 1$.
Our key ingredient will be the ability to compute the next index that will be
sampled, without touching elements that are not sampled.

Let $x_{-i}$ be an element just sampled. We can now define a random variable
$\skipp(i)$ that indicates how many elements past $x_{-i}$ will be skipped over
before selecting index $-(i + \skipp(i))$ according to the distribution given by
\simpleRRAlg{}. Conveniently, this random variable can be efficiently generated
in $O(1)$ time using the inverse transformation
method~\cite{Ross09:prob-model-book} because its cumulative distribution
function (CDF) has a simple, efficiently-solvable form: $\prob{\skipp(i) \leq k}
= 1 - \prod_{t=i+1}^{i+k} (1 - p_{-t}) = 1 - \frac{i}{i+k}= \frac{k}{i+k}$.
This leads to the following improved algorithm:
\begin{algorithm}[h]
  \caption{\fastSingleRRAlg$(X)$ --- Fast RR sampling for $s=1$}
  \KwIn{a stream segment $X = \langle x_1, \dots, x_{|X|} \rangle$.}

  \KwOut{a set $\{(k_i, \ell_i)\}$, where $k_i$ is an index into $X$ and $\ell_i
    = 1$}
  \tcp{The same input/output behaviors as \simpleRRAlg{}.}

  $i \gets 1$\\
  \While{$i < |X|$}{
    $A \gets A \cup \{(i, 1)\}$\\
    $i \gets i + \skipp(i)$\\
  } 
  
  \Return {$A$}
\end{algorithm}

This improvement significantly reduces the number of iterations:
\newcommand*{\FSR}[0]{\normalfont\footnotesize\textsf{FSR}}
\begin{lemma}
  \label{lemma:fsr-count}
  Let $T_{\FSR}(n)$ be the number of times the \textbf{while}-loop in the
  \fastSingleRRAlg{} algorithm is executed on input $X$ with $n = |X|$. Then,
  $\expct{T_{\FSR}(n)} = O(1 + \log(n))$. Also, for $m \geq n$ and $c \geq
  4$, \(\prob{T_{\FSR}(n) \geq 1 + c\cdot \log(m)} \leq m^{-c}\).
\end{lemma}
\begin{proof}
  Let $Z_i$ be an indicator variable for whether $x_{-i}$ contributes to an
  iteration of the \textbf{while}-loop. Hence, $T_{\FSR}(n) = 1 + Z$,
  where $Z = \sum_{i=2}^{|X|} Z_i$. But $\prob{Z_i = 1} = 1/i$, so $\expct{Z} =
  \frac12 + \frac13 + \dots + \frac1n \leq \ln n$.  This proves the expectation
  bound.

  Let $t = c \log_2 m$, so $t \geq 4\cdot \log_2 m > 2e\ln(2)\cdot \log_2m =
  2e\ln m \geq 2e \ln n \geq 2e\expct{Z}$. Therefore, by
  Theorem~\ref{thm:chernoff}, we have that
  \begin{align*}
    \prob{T_{\FSR}(n) \geq 1 + c\cdot \log_2(m)}  \leq \prob{Z \geq c\cdot \log_2 m}
    \leq 2^{-t} =  m^{-c},
  \end{align*}
  which concludes the proof.
\end{proof}

Immediately, this means that if $A =
\fastSingleRRAlg(X)$ is kept as a simple sequence (e.g., an array), the running
time---as well as the length of $A$---will be $O(1 + \log(|X|))$ in expectation.
Moreover, $A$ follows the same distribution as \simpleRRAlg{} with $s = 1$, only
more efficiently computed.

\begin{remark}
  Vitter~\cite{Vitt85} studied a related problem that requires
  sampling from the same distribution as {\simpleRRAlg{}}. He developed a
  sophisticated algorithm (Algorithm Z) for generating the skip offsets for $s
  \geq 1$. Our $\fastSingleRRAlg{}$ algorithm addresses the special case where
  $s = 1$, which is significantly simpler and does not require the same level of
  machinery as Vitter's Algorithm Z.
\end{remark}

\subsection{Improved Multiple-Element Sampler}
In the general case of reversed reservoir sampling, generating skip offsets from
the distribution for $s > 1$ turns out to be significantly more involved than
for $s = 1$. While this is still possible, e.g., using a variant of Vitter's
Algorithm Z~\cite{Vitt85}, prior algorithms appear inherently sequential.

This section describes a new parallel algorithm that builds on
\fastSingleRRAlg{}. In broad strokes, it first ``oversamples'' using a simpler
distribution and subsequently, ``downsamples'' to correct the sampling
probability. To enable parallelism, we logically divide the stream segment into
$s$ ``tracks'' of roughly the same size and have the single-element algorithm work on
each track in parallel. 

\paragraph{Division of Work via Tracks.} The aim of the first phase is to sample
each element $x_{-i}$ with a slightly higher probability than $p_{-i}^{(s)}$, in
a way that results in about $s + s\log(|X|/s)$ elements sampled at the end of
the phase---and the sampling could be carried out in parallel, with depth less
than $s$. To this end, we logically divide the stream segment into $s$ tracks of
about the same size and in parallel, have the single-element algorithm work on
each track.

\paragraph{Track View.}
Define $\createViewAlg(X, k)$ to return a view corresponding to track $k$ on
$X$: if $Y = \createViewAlg(X, k)$, then $Y_{-i}$ is $X[-\alpha_s^{(k)}(i)]$,
where $\alpha_s^{(k)}(i) = i\cdot{}s + k$.
That is, track $k$ contains, in reverse order, indices $-(s+k), -(2s+k),
-(3s+k), \dots$. Importantly, these views never have to be materialized.

\begin{algorithm}[tbh]
  \caption{\fastRRAlg$(X, s)$ --- Fast reversed reservoir sampling}
  \label{algo:fastRR}
  \KwIn{a stream segment $X = \langle x_1, \dots, x_{|X|} \rangle$ and a
    parameter $s > 0$, $s \leq |X|$.}

  \KwOut{a set $\{(k_i, \ell_i)\}$, where $k_i$ is an index into $X$ and $\ell_i
    \in [s]$}

  $\pi \gets $ draw a random permutation of $[s]$ \\
  $T_0 \gets \{ (i, \pi_i) \mid i = 1, 2, \dots, s\}$\\
  \For{$\tau = 1, 2, \dots, s$ \textbf{\upshape in parallel}}{
    $\subsegX^{(\tau)} \gets \createViewAlg(X, \tau)$ \\
    $T_\tau \gets \fastSingleRRAlg(\subsegX^{(\tau)})$ \\
    $T'_\tau \gets $ $\{ (i, \ell) \in T_\tau \mid$ \coinFlip$({i\cdot
      s}/{\alpha_s^{(\tau)}(i)}) = $~\coinHeads{}$\}$ \tcp{filter, keep if
      coin shows heads}

    $T''_\tau \gets \{(i,\,$\uniformSample$(1, s)) \mid (i, \text{\texttt{\string_})} \in T'_\tau \}$ \tcp{map}

  }

  \Return {$T_0 \cup T''_1 \cup T''_2 \cup \dots \cup T''_s$}
\end{algorithm}

Algorithm~\ref{algo:fastRR} combines the ideas developed so far. We now argue
that \fastRRAlg{} yields the same distribution as \simpleRRAlg{}:

\begin{lemma}
  Let $A$ be a return result of $\fastRRAlg(X, s)$. Then, for $j=1,\dots, |X|$
  and $\ell \in [s]$, \( \prob{(j, \ell) \in A} = \frac{1}{s}\cdot
  p^{(s)}_{-j}. \)
\end{lemma}
\begin{proof}
  For $j \leq s$, age $j$ is paired with a slot $\ell$ drawn from a random
  permutation of $[s]$, so $\prob{(j, \ell) \in A} = \frac{1}{s} = \frac{1}{s}
  \cdot 1 = \frac{1}{s}\cdot p^{(s)}_{-j}$.
  For $j > s$, write $j$ as $j = s\cdot i + \tau$, so age $j$ appears as age $i$
  in view $\subsegX^{(\tau)}$. Now age $j$ appears in $A$ if both of these
  events happen: (1) age $i$ was chosen into $T_\tau$ and (2) the coin turned
  up heads so it was retained in $T'_\tau$. These two independent events happen
  together with probability \[ p^{(1)}_{-i} \cdot \frac{i\cdot
    s}{\alpha_s^{(\tau)}(i)} = \frac{1}{i}\cdot \frac{i\cdot s}{s\cdot i + \tau}
  = \frac{s}{j} = p^{(s)}_{-j}.\] Once age $j$ is chosen, it goes to slot
  $\ell$ with probability $1/s$. Hence, $\prob{(j, \ell) \in A} =
  \frac{1}{s}\cdot p^{(s)}_{-j}$.
\end{proof}

Because the cost of an algorithm depends on the choice of data structures, we
defer the cost analysis of \fastRRAlg{} to the next section, after we discuss
how the reserved samples will be stored.

\subsection{Storing and Retrieving Reserved Samples}
\emph{How should we store the sampled elements?} An important design goal is for
samples of any size $q \leq s$ to be generated without first generating $s$
samples. To this end, observe that restricting $\drawPerm(A)$ to its first $q
\leq s$ coordinates yields a $q$-permutation over the input. This motivates a
data structure that stores the contents of different slots separately.

Denote by $\rsStore(A)$, or simply $\rsStore$ in clear context, the
\emph{binned-sample} data structure for storing reserved samples $A$. The
samples are organized by their slot numbers $(\rsStore_i)_{i=1}^s$, with
$\rsStore_i$ storing slot $i$'s samples. Within each slot, samples are binned by
their ages. In particular, each $\rsStore_i$ contains $\lceil \log_2(\lceil
|X|/s \rceil) \rceil + 1$ bins, numbered $0, 1, 2, \dots, \lceil \log_2(\lceil
|X|/s \rceil) \rceil$---with bin $k$ storing ages $j$ in the range
$2^{k-1} < \lceil j/s \rceil \leq 2^k$. Below, bin $t$ of slot $i$ will be
denoted by $\rsStore_i[t]$.

Additional information is kept in each bin for fast queries: every bin $k$
stores $\phi(k)$, defined to be the age of the oldest element in bin $k$ and all
younger bins for the same slot number.

Below is an example. Use $s = 3$ and $|X| = 16$. Let the result from
\fastRRAlg{} be \[ A = \{(1, 2), (2, 3), (3, 1), (7, 1), (10, 3), (11, 3),
(14, 2)\}.\] Then, $\rsStore$ keeps the following bins, together with $\phi$
values:
\begin{center}
  \small
  \setlength{\tabcolsep}{10pt}
  \begin{tabular}{r c c c c}
    \textsf{\bfseries Bin}:& $\rsStore_i[0]$ & $\rsStore_i[1]$ & $\rsStore_i[2]$ & $\rsStore_i[3]$ \\
    \midrule
    Slot $i = 1$ & $\{3\}_{\phi=3}$ & $\emptyset_{\phi=3}$ & $\{7\}_{\phi=7}$ & $\emptyset_{\phi=7}$ \\
    Slot $i = 2$ & $\{1\}_{\phi=1}$ & $\emptyset_{\phi=1}$ & $\emptyset_{\phi=1}$ & $\{14\}_{\phi=14}$ \\
    Slot $i = 3$ & $\{2\}_{\phi=2}$ & $\emptyset_{\phi=2}$ & $\{10, 11\}_{\phi=11}$ & $\emptyset_{\phi=11}$\\
    \bottomrule
  \end{tabular}
\end{center}

From this construction, the following claims can be made:
% \begin{lemma}
%   \label{claim:size-of-bin}
%   \label{claim:size-of-slot}
%   \begin{enumerate}[label=(\roman*)]
%   \item The expected size of the bin $\rsStore_i[t]$ is $\expct{|\rsStore_i[t]|}
%     \leq 1$.
%   \item The size of slot $\rsStore_i$ is expected $O(1 + \log(|X|/s))$.
%     Furthermore, for $c \geq 4$, $\prob{|\rsStore_i| \leq 1 + c\log_2(|X|)} \geq
%     1 - |X|^{-c}$.
% \end{enumerate}
% \end{lemma}
% (The proof appears in the full version~\cite{our_full_paper}. Very briefly, the
% first is true because a bin $t$ covers $\gamma = s \cdot 2^{t-1}$ indices, each
% sampled into slot $i$ w.p. at most $1/\gamma$; the second follows from an
% argument similar to Lemma~\ref{lemma:fsr-count}.)

\begin{claim}
  \label{claim:size-of-bin}
  The expected size of the bin $\rsStore_i[t]$ is $\expct{|\rsStore_i[t]|} \leq
  1$.
\end{claim}
\begin{proof}
  Bin $t$ of $\rsStore_i$ is responsible for negative indices $j$ in the range
  $2^{t-1} < \lceil -j/s \rceil \leq 2^t$, for a total of $s(2^t - 2^{t-1}) =
  s\cdot 2^{t-1}$ indices. Among these indices, the index that has the highest
  probability of being sampled is $-(s2^{t-1} + 1)$, which is sampled into slot $i$
  with probability $\frac{1}{s}\cdot\frac{s}{s2^{t-1} + 1} \leq
  \frac{1}{s\cdot2^{t-1}}$.  Therefore,
  \[
    \expct{|\rsStore_i[t]|} \leq s\cdot 2^{t-1} \cdot\frac{1}{s\cdot2^{t-1}} = 1,
  \]
  which concludes the proof.
\end{proof}

\begin{claim}
  \label{claim:size-of-slot}
  The size of slot $\rsStore_i$ is expected $O(1 + \log(|X|/s))$. Furthermore, for
  $c \geq 4$, $\prob{|\rsStore_i| \leq 1 + c\log_2(|X|)} \geq 1 - |X|^{-c}$.
\end{claim}
\begin{proof}
  Let $Y_t = \onef{x_{-t}\text{ is chosen into slot $i$}}$, so $|\rsStore_i| =
  \sum_{t=1}^{|X|} Y_t$. Since $\expct{Y_t} = p^{(s)}_{-t}/s =
  \frac{1}{s}\min(1, s/t)$, we have
  \begin{align*}
    \expct{|\rsStore_i|} = \sum_{t=1}^{|X|}
    \expct{Y_t} = 1 + \sum_{t=s+1}^{|X|} \frac{1}{t} \leq 1 + \int_{t=s}^{|X|}
    \frac{dt}{t} = 1 + \ln\left(\tfrac{|X|}{s}\right),
  \end{align*}
  which proves the expectation bound. Because $Y_t$'s are independent, using an
  argument similar to the proof of Lemma~\ref{lemma:fsr-count}, we have the
  probability bound.
\end{proof}

\paragraph{Data Structuring Operations.} 
Algorithm~\ref{algo:build-and-query} shows algorithms for constructing a
binned-sample data structure and answering queries. To \makeAlg{} a
binned-sample data structure, the algorithm first arranges the entries into
groups by slot number, using a parallel semisorting algorithm, which reorders an
input sequence of keys so that like sorting, equal keys are arranged
contiguously, but unlike sorting, different keys are not necessarily in sorted
order. Parallel semisorting of $n$ elements can be achieved using $O(n)$
expected work and space, and $O(\log n)$ depth~\cite{GuShunSunBlelloch15}. It
then, in parallel, processes each slot, putting every entry into the right bin.
Moreover, it computes a $\min$-prefix, yielding $\phi(\cdot)$ for all bins.
There is not much computation within a slot, so we do it sequentially but the
different slots are done in parallel. To answer a \queryAlg{} query, the
algorithm computes, for each slot $i$, the oldest age within $X[-w\texttt{:}]$
that was assigned to slot $i$. This can be found quickly by figuring out the bin
$k$ where $w$ should be. Once this is known, it simply has to look at $\phi$ of
bin $k - 1$ and go through the entries in bin $k$. This means a query touches at
most two bins per slot.

\begin{algorithm}[tbh]
  \caption{Construction of binned-sample data structure and query}
  \label{algo:build-and-query}

  \SetKwProg{Fn}{}{:}{}

  \tcp{Below, use the convention that $\max \emptyset = -\infty$}

  \Fn{\makeAlg$(A, n, s)$}{
    \KwIn{$A$ is a sequence of reserved samples, $n$ is the length of the
      underlying stream segment $X$, and $s$ is the target sample size used to
      generate $A$.}
    \KwOut{an instance of binned-sample structure $\rsStore(A)$}
    Use semisorting to arrange $A$ into $G_1, G_2, \dots, G_s$ by
    slot number \\
    \For{$i = 1, \dots, s$~\textbf{\upshape in parallel}}{
      Create bins $\rsStore_i[0], \dots, \rsStore_i[\beta]$, $\beta =
      \lceil \log_2(\lceil n/s \rceil) \rceil$\\
      \ForEach{$(j, \texttt{\string_}) \in G_i$}{
        Write $j$ into $\rsStore_i[k]$, where
        $2^{k-1} < \lceil j/s \rceil \leq 2^k$ }
      Let $\phi(\rsStore_i[0]) = \max \rsStore_i[0]$\\
      \tcp{prefix max}
      \For{$k = 1, \dots, \beta$}{
        $\phi(\rsStore_i[k]) \gets \max(\phi(\rsStore_i[k-1]), \max \rsStore_i[k])$
      }
    } 
    \Return{$\rsStore$}
  }
  \Fn{\queryAlg$(\rsStore, q, w)$}{
    \KwIn{$\rsStore$ is a binned-sample structure, $q$ is the number of samples
      desired, $w$ tells the algorithm to draw sample from $X[-w:]$.}
    \KwOut{a $q$-permutation of $X[-w:]$}
    \For{$i = 1, \dots, q$~\textbf{\upshape in parallel}}{
      Let $k$ be such that $2^{k-1} < \lceil w/s \rceil \leq 2^k$\\
      $\gamma \gets \max\{ j \in \rsStore_i[k] \mid j \leq w \}$ \tcp{The oldest
        that is at least as young as $w$}
      $r_i \gets \max(\gamma, \phi(\rsStore_i[k-1]))$
    }

    \Return{$(r_1, r_2, \dots, r_q)$}
  }
\end{algorithm}

\paragraph{Cost Analysis.} We now analyze \fastRRAlg{}, \makeAlg{}, and
\queryAlg{} for their work and parallel depth.  More concretely, the following
claims are made:

\begin{claim}
  By storing $T_0, T_i$'s, and $T'_i$'s as simple arrays, \fastRRAlg$(X, s)$
  runs in expected $O(s + s\log \tfrac{|X|}s)$ work and $O(\log |X|)$ parallel
  depth.
\end{claim}
\begin{proof}
  Generating a permutation of $[s]$ can be done in $O(s)$ work and $O(\log s)
  \leq O(\log |X|)$ depth. There are $s$ tracks, each, in parallel, calling
  \fastSingleRRAlg{}, which is a sequential algorithm. Once $T_\tau$ is known,
  the remaining steps are simple \lstinline{map} and \lstinline{filter}
  operations, which have $O(|T_\tau|)$ work and $O(O(|T_\tau|) \leq O(\log |X|)$
  depth. Therefore, the dominant cost is \fastSingleRRAlg{}. By
  Lemma~\ref{lemma:fsr-count}, each track performs $O(1 + \ln(|X|/s))$ work in
  expectation. Summing across $s$ tracks, the total work is expected $O(s +
  s\log(|X|/s))$. In terms of parallel depth, by Lemma~\ref{lemma:fsr-count},
  each track finishes after $1 + 4\log_2(|X|)$ iterations with probability at
  least $1 - |X|^{-4}$. Applying the union bound, we have that the expected
  depth overall is at most $O(\log |X|)$ provided that $|X| \geq s$.
\end{proof}

\begin{claim}
  \makeAlg$(A, n, s)$ runs in $O(s + s\log \tfrac{n}{s})$ work and $O(\log n)$
  parallel depth.
\end{claim}
\begin{proof}
  Let $\psi = s + O(s\log(n/s))$. The algorithm starts with a semisorting step,
  which takes $O(|A|)$ work and $O(\log |A|)$ depth to arrange the entries of
  $A$ into $G_1, \dots G_2$. Since $|A|$ is expected $\psi$ but never exceeds
  $n$. This step takes $O(\psi)$ work and $O(\log n)$ depth. For each $i = 1,
  \dots, s$, the size of $G_i$ is expected $\frac{1}{s}\psi$
  (Claim~\ref{claim:size-of-slot}). Therefore, the work performed for each $i$
  is expected $\frac{1}{s}\psi$, for a total of $\psi$ across all $s$ slots in
  expectation. Because by Claim~\ref{claim:size-of-slot}, the size of a $G_i$
  exceeds $1 + 4\log_2(n)$ with probability at most $1/n^4$. The overall depth
  is at most $O(\log n)$ \whp.
\end{proof}
\begin{lemma}
  \label{lemma:cost-of-query}
  \queryAlg$(\rsStore, q, t)$ runs in $O(q)$ work and $O(\log n)$ parallel
  depth, where $n$ is the length of $X$ on which $\rsStore$ was built.
\end{lemma}
\begin{proof}
  The algorithm looks into $q$ slots in parallel. For each slot $\rsStore_i$, it
  takes $T_i = O(1 + |\rsStore_i[k]|)$ sequential time, which is expected $O(1)$
  by Claim~\ref{claim:size-of-bin}. Hence, the total work is expected $O(s)$.
  Then, it follows from Claim~\ref{claim:size-of-slot} that $T_i$ exceeds $1 +
  4\log_2(n)$ with probability at most $1/n^4$, so the overall depth is at most
  $O(\log n)$ \whp.
\end{proof}

\subsection{Handling Minibatch Arrival}
This section describes how to incorporate a minibatch into our data structure to
maintain a sliding window of size $W$. Assume that the minibatch size is $n_i
\le W$. If not, we can only consider its $W$ most recent elements. When a
minibatch arrives, retired sampled elements must be removed and the remaining
sampled elements are ``downsampled'' to maintain the correct distribution.

Remember that the number of selected elements is $O(s + s\log (W/s))$ in
expectation, so we have enough budget in the work bound to make a pass over them
to filter out retired elements. Instead of revisiting every element of the
window, we apply the process below to the selected elements to maintain the
correct distribution. Notice that an element at age $i$ was sampled into slot
$\ell$ with probability $\frac{1}{s}p^{(s)}_{-i}$. A new minibatch will cause
this element to shift to age $j$, $j > i$, in the window. At age $j$, an element
is sampled into slot $\ell$ with probability $\frac{1}{s}p^{(s)}_{-j}$. To
correct for this, we flip a coin that turns up heads with probability
${p^{(s)}_{-j}}/{p^{(s)}_{-i}} \leq 1$ and retain this sample only if the coin
comes up heads.

Therefore, \lstinline{insert}$(B_i)$, $|B_i| = n_i$ handles a minibatch arrival
as follows:
\begin{enumerate}[label=Step~\roman*:, leftmargin=2.5\parindent, topsep=2pt]
\item Discard and downsample elements (above) in $\rsStore$; the index shifts by
  $n_i$.
\item Apply \fastRRAlg{} on $B_i$, truncated to the last $W$ elements if $n_i >
  W$.
\item Run \makeAlg{} on the result of \fastRRAlg{} with a modification where the
  it appends to an existing $\rsStore$ as opposed to creating a new structure.
\end{enumerate}

Overall, this leads to the following cost bound for \lstinline{insert}:
\begin{lemma}
  \label{lemma:cost-of-insert}
  \lstinline{insert} takes $O(s + s\log(W/s))$ work and $O(\log W)$ depth.
\end{lemma}
\begin{proof}
  The cost of discarding and downsampling elements in the existing $\rsStore$ is
  no more than the cost of running \makeAlg{} with $n = W$ because the cost of
  the former is upper-bounded by the cost of going through every element in
  $\rsStore$ once. Hence, this step takes $O(s + s\log(W/s))$ work and $O(\log
  W)$ depth. When \fastRRAlg{} is run, it is run with input length $\min(W,
  n_i)$. Thus, it performs no more work than allowed by the lemma. Finally,
  \makeAlg{} is called with $n \leq W$, costing $O(s + s\log(W/s))$ work and
  $O(\log W)$ depth, which is also the total cost of \lstinline{insert}.
\end{proof}

\subsection{Total Cost of \sworvar{}}
The cost of handling a minibatch's arrival is given by
Lemma~\ref{lemma:cost-of-insert}, and the cost of answering a query is given by
Lemma~\ref{lemma:cost-of-query}. Altogether, this proves
Theorem~\ref{thm:main-var-win}.
Furthermore, across $t$ minibatches, the total work is
\[
  O\left( t + \sum_{i=1}^t (s + s\ln(W/s)) \right).
\]

%%% Local Variables:
%%% mode: latex
%%% TeX-master: "paper"
%%% End:

%---------------------------------
\section{Parallel Sampling from an Infinite Window}
\label{sec:infinite}
%---------------------------------
\newcommand{\expec}{\expct}

This section addresses sampling without replacement from the infinite window,
which includes all elements seen so far in the stream. This is formulated as the
\sworinf~task: For each time $i = 1, \ldots, t$, maintain a random sample of size
$\min\{s, N_i\}$ chosen uniformly without replacement from $\stream_i$.

While \sworinf~can be reduced to the sliding-window setting, by setting the
window size to the number of elements received so far, in this section, we show
that there is an algorithm for infinite window that is simpler and more
efficient. We further show it to be work optimal, up to constant factors.

For $p,q \in [r]$, let $\hyper(p,q,r)$ be the \emph{hypergeometric random
  variable}, which can take an integer value from $0$ to $\min\{p,q\}$. Suppose
there are $q$ balls of type 1 and $(r-q)$ balls of type 2 in an urn. Then,
$\hyper(p,q,r)$ is the number of balls of type $1$ drawn in $p$ trials, where in
each trial, a ball is drawn at random from the urn without replacement. It is
known that $\expct{\hyper(p,q,r)} = \frac{pq}{r}$.

%---------------------------------------
\paragraph{Work Lower Bound.}
%---------------------------------------------
We first show a lower bound on the work of any algorithm for \sworinf,
sequential or parallel, by considering the expected change in the sample output
after a new minibatch is received.
%-------------------
\begin{lemma}
\label{lem:sworinf-lb}
Any algorithm that solves \sworinf~ must have expected work at least
\[\Omega\left(t + \sum_{i=1}^t \min\left\{n_i, \frac{sn_i}{N_i}\right\}\right)\]
over $t$ timesteps.
\end{lemma}
%-------------------

%-------------------
\begin{proof}
  First consider the number of elements that are sampled from each minibatch. If
  $N_i \le s$, then the entire minibatch is sampled, resulting in a work of
  $\Omega(n_i)$. Otherwise, the number of elements sampled from the new
  minibatch $B_i$ is $\hyper(s, n_i, N_i)$. The expectation is
  $\expec{\hyper(s,n_i,N_i)} = \frac{s \cdot n_i}{N_i}$, which is a lower bound
  on the expected cost of processing the minibatch. Next, note that any
  algorithm must pay $\Omega(1)$ for examining minibatch $B_i$, since in our
  model the size of the minibatch is not known in advance. If an algorithm does
  not examine a minibatch, then the size of the minibatch may be as large as
  $\Omega(N_i)$, causing $\Omega(1)$ elements to be sampled from it. The
  algorithm needs to pay at least $\Omega(t)$ over $t$ minibatches. Hence, the
  total expected work of any algorithm for \sworinf~ after $t$ steps must be
  $\Omega\left(t + \sum_{i=1}^t \min\{n_i, \frac{sn_i}{N_i}\}\right)$.
\end{proof}
%-------------------

%---------------------------------------------
\paragraph{Sequential Algorithm for \sworinf.}
%---------------------------------------------
We present a simple sequential algorithm for \sworinf, whose work matches the lower bound from Lemma~\ref{lem:sworinf-lb}. It uses a subroutine for sampling without replacement from a static array. 

%----------------------------
\begin{observation}[\cite{Ahrens85,Vitt87}]
\label{obs:seq-swor}
There is an algorithm for choosing a random sample of size $s$ without replacement from a (static) array of size $r$ using $O(s)$ work. 
\end{observation}
%----------------------------

The idea in a solution to \sworinf, described in Algorithm~\ref{algo:sworinf-seq}, is as follows. When a minibatch $B_i$ arrives, a random variable $\kappa$ is generated according to the hypergeometric distribution to determine how many of the $s$ samples need to be chosen from $B_i$, rather than minibatches that arrived before $B_i$. The algorithm then chooses a random sample of size $\kappa$ without replacement from $B_i$, and updates the sample $S$ accordingly.

%-----------------------
\begin{algorithm}
Initialization: Sample $S \gets \emptyset$\;
\tcp{$n_i = |B_i|$ and $N_i = \sum_{j=1}^i n_i$}
\If{minibatch $B_i$ is received}
{
   \If{$N_i \le s$}    
   {
     $S \gets S \cup B_i$  \tcp{Store the entire minibatch}
   }
   \Else{
       Let $\kappa$ be a random number drawn from $\hyper(s, n_i, N_i)$\;
       
       Let $S_i$ be a set of $\kappa$ elements sampled without replacement from $B_i$\; %, using Observation~\ref{obs:seq-swor}\;
       
       Replace $\kappa$ randomly chosen elements in $S$ with $S_i$.\;
   }
}
\caption{Work-Optimal Sequential Algorithm for \sworinf.}
\label{algo:sworinf-seq}
\end{algorithm}
%-----------------------

%-----------------------
\begin{lemma}
\label{lem:sworinf-seq}
Algorithm~\ref{algo:sworinf-seq} is a sequential solution to \sworinf~ with work $O\left(t + \sum_{i=1}^t \min\{n_i, \frac{sn_i}{N_i}\} \right)$ to process $t$ minibatches $B_1,\ldots,B_t$. This work is optimal given the lower bound from Lemma~\ref{lem:sworinf-lb}.
\end{lemma}
%-----------------------
% This holds because if minibatch $B_i$ ends up replacing $\kappa_i$ elements, the
% cost is $O(1+\kappa_i)$ with $\kappa_i \leq n_i$, but then $\expct{\kappa_i} =
% s\cdot{}n_i/N_i$. See the full paper for detail~\cite{our_full_paper}.

%-----------------------
\begin{proof}
If $N_i \le s$, then the arriving minibatch $B_i$ is added to the sample in its entirety, for a total work of $\Theta(n_i)$. If $N_i > s$, then the algorithm has to pay for (1)~generating a random variable according to the hypergeometric distribution, which can be done in $O(1)$ time by an algorithm such as described in~\cite{Stadlober90}, (2)~selecting a sample of size $\kappa \le n_i$ from $B_i$, which can be done in $O(\kappa)$ time, from Observation~\ref{obs:seq-swor}, and (3)~replacing $\kappa$ randomly chosen elements from $S$ -- this can be done in $O(\kappa)$ time by choosing $\kappa$ random elements without replacement from $S$ using Observation~\ref{obs:seq-swor}, and overwriting these locations with the new samples. The overall time for processing this batch is $O(1+ \kappa) = O(1+ \hyper(s,n_i,N_i))$. The expected time for processing the minibatch is $O(\expec{1+ \hyper(s,n_i,N_i)}) = O(1+ \frac{s \cdot n_i}{N_i})$. Overall, the cost of processing $B_i$ is $O(1+ \min\{n_i, \frac{s \cdot n_i}{N_i}\})$, and the total work of the algorithm is $O\left(t + \sum_{i=1}^t \min\{n_i, \frac{sn_i}{N_i}\} \right)$. 
\end{proof}
%-----------------------

%-------------------------------------------------------
\paragraph{Parallel Algorithm for \sworinf.}
\label{sec:sworinf-par}
%------------------------------------------------------
Our solution is presented in Algorithm~\ref{algo:sworinf-par}. The main idea is
as follows: When a minibatch $B_i$ arrives, generate a random variable $\kappa$
according to the hypergeometric distribution to determine how many of the $s$
samples will be chosen from $B_i$, as opposed to prior minibatches. Then,
choose a random sample of size $\kappa$ without replacement from
$B_i$ and update the sample $S$ accordingly.
%Based on the sequential algorithm described above, we derive a parallel
%algorithm for \sworinf.
We leverage Sanders et~al.~\cite{Sanders18}'s recent algorithm for parallel
sampling without replacement (from static data), restated below in
the work-depth model:

%--------------------
\begin{observation}[\cite{Sanders18}]
\label{obs:swor-par}
There is a parallel algorithm to draw $s$ elements at random without replacement from $N$ elements using $O(s)$ work and $O(\log s)$ depth.
\end{observation}
%--------------------

%------------------------------------------------------
%\begin{observation}[\cite{Sanders18}]
%There is a $p$ processor parallel algorithm for uniformly sampling $s$ elements without replacement from $N$ elements using time $O(s/p + \log p)$ per processor.
%\end{observation}
%------------------------------------------------------

%While the algorithm in~\cite{Sanders18} is not stated in the work-depth model, we can arrive at the following tradeoffs. Note that the total work on $p$ processors is $O(s + p \log p)$. As long as $p = O(s/\log s)$, the total work is $O(s)$, making the algorithm work-efficient. By setting $p = \Theta(s/\log s)$, we can arrive at the following corollary.

%Combining the above result with Algorithm~\ref{algo:sworinf-seq}, we derive a parallel algorithm for \sworinf, through parallelizing the individual steps of the algorithm. 
Our algorithm uses static parallel sampling without replacement in two places:
once to sample new elements from the new minibatch, and then again to update the
current sample. In more detail, when a minibatch arrives, the algorithm
\textbf{(i)}~chooses $\kappa$, the number of elements to be sampled from $B_i$,
in $O(1)$ time; \textbf{(ii)}~samples $\kappa$ elements without replacement from
$B_i$ in parallel; and \textbf{(iii)}~replaces $\kappa$ randomly chosen elements
in $S$ with the new samples using a two-step process, by first choosing the
locations in $S$ to be replaced, followed by writing the new samples to the
chosen locations. Details appear in Algorithm~\ref{algo:sworinf-par}.

%-----------------------
\begin{algorithm}
Initialization: Sample $S \gets \emptyset$\;

\If{minibatch $B_i$ is received}
{
   \tcp{Recall $n_i = |B_i|$ and $N_i = \sum_{j=1}^i n_i$}
   \lIf{$N_i \le s$}{
     Copy $B_i$ into $S$ in parallel 
   }%\tcp*{\normalfont\itshape the whole minibatch}
   \Else{
       Let $\kappa$ be a random number generated by $\hyper(s, n_i, N_i)$\;
       
       $S_i \gets$ $\kappa$ elements sampled without replacement from $B_i$ (Obs.~\ref{obs:swor-par})\;
       
       $R_i \gets$ $\kappa$ elements sampled without replacement from $\{1,\ldots,s\}$ (Obs.~\ref{obs:swor-par})\;
       
       \lFor{$j = 1$ to $\kappa$}{Replace $S[R_i[j]] \gets S_i[j]$}
   }
}
\caption{Parallel Algorithm for \sworinf.}
\label{algo:sworinf-par}
\end{algorithm}
%-----------------------

%-----------------------
\begin{theorem}
%\label{lem:sworinf-par}
\label{thm:sworinf-par}
Algorithm~\ref{algo:sworinf-par} is a work-efficient algorithm for \sworinf. The
total work to process $t$ minibatches $B_1,\ldots,B_t$ is $O\left(t +
  \sum_{i=1}^t \min\{n_i, \frac{sn_i}{N_i}\} \right)$ and the parallel depth of
the algorithm for processing a single minibatch is $O(\log s)$. This work is
optimal up to constant factors, given the lower bound from
Lemma~\ref{lem:sworinf-lb}.
% \end{lemma}
\end{theorem}
% -----------------------

%-----------------------
\begin{proof}
When a new minibatch $B_i$ arrives, for the case case $N_i \le s$, copying $n_i$ elements from $B_i$ to $S$ can be done in parallel in $O(n_i)$ work and $O(1)$ depth, by organizing array $S$ so that the empty locations in the array are all contiguous, so that the destination for writing an element can be computed in $O(1)$ time.

For the case $N_i > s$, random variable $\kappa$ can be generated in $O(1)$ work. The next two steps of sampling $\kappa$ elements from $B_i$ and from $\{1,\ldots,n\}$ can each be done using $O(\kappa)$ work and $O(\log \kappa)$ depth, using Observation~\ref{obs:swor-par}. The final for loop of copying data can be performed in $O(\kappa)$ work and $O(1)$ depth. Hence, the expected total work for processing $B_i$ is $1 + \min\{n_i, \frac{sn_i}{N_i}\}$, and the depth is $O(\log \kappa)$. Added up over all the $t$ minibatches, we get our result. Since $\kappa \le s$, the parallel depth is $O(\log s)$. 
\end{proof}
%-----------------------

%%% Local Variables:
%%% mode: latex
%%% TeX-master: "paper"
%%% End:

%--------------------------------------------------
\subsection{Fixed Length Sliding Window}
\label{sec:fixed}
%--------------------
In the {\bf fixed length sliding window} model with window size $w$, the sample is desired from the $w$ most recent elements, where the window size $w \gg s$ is known in advance. While it is more restrictive than the variable length sliding model, this model is relevant when the aggregation function and their parameters are known in advance, as is often the case in a streaming system with long-lived continuous queries. Algorithms for this model are simpler than those for variable length sliding window.\\

%Let $\stream_i^$ denote the union of the $w$ most recent minibatches, i.e. $\stream_i^w = \cup_{j=(i-w+1)}^i \stream_j$. For $i<w$, $\stream_i^w = \stream_i$. Let $N_i^w$ denote the size of $\stream_i^w$.\\

\noindent {\bf Task \sworfixed:} For each time $i = 1 \ldots t$, after observing minibatch $B_i$, maintain a random sample of size $\min\{s, N_i\}$ chosen uniformly without replacement from the $w$ most recent elements in the stream.
%Thus, for $i > 0$, the $i$th bucket will contain minibatches indexed from $(i-1) \cdot w$ till $i \cdot w$.
%Note that even though the number of minibatches within a window is fixed to $w$, the number of elements within a window are not fixed. This prevents us from applying other techniques such as those of~\cite{BOZ09}, who take advantage of the fact that the number of elements in a window is always a number that is known beforehand.
We first present a lower bound on the work of any algorithm for \sworfixed.
%------------------
\begin{lemma}
For any algorithm for \sworfixed, the expected work to process minibatch $i$ is at least $\Omega\left(1+ \min\left\{s, \frac{s \cdot n_i}{w}\right\} \right)$. The expected work to process $t$ minibatches is at least $\Omega\left(t + s \sum_{i=1}^t \min\left\{1, \frac{n_i}{w}\right\}\right)$.
%$\Omega\left(t + s\left(1+ N_t/w \right)\right))$.
\end{lemma}
%------------------

%------------------
\begin{proof}
After minibatch $B_i$ is received, the random sample must contain $s$ elements randomly chosen from the $w$ most recent elements in $\stream_i$. 
In expectation, the number of elements that will be chosen from $B_i$ is
(1)~If $n_i \ge w$, then $s$, since the entire window is contained within mini-batch $B_i$, and
(2)~If $n_i < w$, then $s \frac{n_i}{w}$, since $s$ elements are chosen out of a total of $w$ elements.
Since the batch size $n_i$ is not known in advance, the work to process the new batch is at least $\Omega(1)$, since the batch has to be examined. If the batch were not examined, it is possible that $n_i \ge w$, and the sample will no longer be correct.
The work to process $B_i$ is at least $\Omega(1 + \min\left\{s, \frac{s \cdot n_i}{w}\right\})$, and the result for $t$ batches follows. \qed
%In the course of processing $t$ batches, the total work is at least $\Omega(t + s \sum_{i=1}^t \min\left\{s, \frac{s \cdot n_i}{w}\right\})$.
%(1)~if $N_i \le s$ then $n_i$, since all elements from $B_i$ are sampled (2)~if $n_i < s < N_i$, then $s \frac{n_i}{w}$, and (3)~if $s \le n_i$, then $s$. The expected change in the random sample, which is $\min\{n_i, s \frac{n_i}{w}, s\}$, is a lower bound on expected work.
%The work to process $B_i$ must be at least $\Omega(1)$, since $B_i$ has to be examined. 
%In the course of processing a sequence of $w$ consecutive elements, the total work is at least $s$, since the sample changes by $s$ elements.
%Over $t$ minibatches, the total number of elements observed is $N_t$, so the total work is $\Omega(s N_t/w)$. The work is also $\Omega(t)$ since each minibatch  requires $\Omega(1)$ time. Further, in the case when $N_t < w$, $\Omega(s)$ work is still needed, since the algorithm must maintain a sample of size $s$. \qed
\end{proof}
%------------------

\paragraph{Algorithm for \sworfixed:} We first consider a sequential algorithm for \sworfixed, which follows by combining our algorithm for \sworinf~with a reduction (due to Braverman et al.~\cite{BOZ09}) from the problem of sampling from a fixed-size sliding window to the problem of sampling from an infinite window. 

This reduction is based on partitioning the stream into ``buckets'', each with a fixed starting and ending point. Each bucket has a width of $w$ elements, where $w$ is the size of the sliding window. For $i=1,2,\ldots$, the $i$th bucket consists of the stream elements at positions $((i-1)\cdot w+1)$ till $i\cdot w$. Note that buckets are distinct from minibatches, which could be of arbitrary sizes. A bucket is also different from a query window, whose boundary need not coincide with a bucket boundary. The algorithm maintains a sample without replacement of size $s$ within each bucket $i$, using the infinite window algorithm. Each per-bucket sample is completed when the bucket ends, and a new sample is started for the next bucket. A bucket that overlaps with the current sliding window is called ``active''. Only active buckets are retained and the rest are discarded. Clearly, there are no more than two active buckets at any time. We call the older of the active buckets as the ``old'' bucket (\oldbkt) and the more recent active bucket as the ``new'' bucket (\newbkt). It is possible that there is no old bucket, if the current window is exactly aligned with the new bucket. We have a sample of size $s$ from \oldbkt, called \oldsmp, and a sample of size $s$ from \newbkt, called \newsmp. Except for corner cases when the new bucket has seen less than $s$ elements, we have that the sizes of both \oldsmp~ and \newsmp~ equal $s$. When a sample is desired from the sliding window, there are two cases. (1)~If the window overlaps only with \newbkt, then return \newsmp. (2)~If the window overlaps partially with \oldbkt~ and partially with \newbkt, then the algorithm selects as many elements from \oldbkt that are not expired yet -- say this is of size $s' \le s$. It then selects the remaining $(s-s')$ elements by (uniformly) subsampling from \newsmp. It is shown in~\cite{BOZ09} that this constitutes a random sample chosen without replacement from a sliding window of size $w$. 

\paragraph{Parallel Algorithm for \sworfixed.}
We now show how to implement this reduction in a parallel algorithm. We note that the focus of \cite{BOZ09} was on the space complexity, while we are interested in the time complexity and work of the sampling algorithm. For maintaining the sample without replacement of size $s$ within each bucket, we use one instance of the parallel algorithm for \sworinf~(Algorithm~\ref{algo:sworinf-par}) for each bucket. When queried, the algorithm uses $O(s)$ work to combine the two sampled.

\begin{lemma}
\label{lem:sworfixed-par}
There is a parallel algorithm for \sworfixed~ whose total work for processing $t$ minibatches is $O\left(t + s \left(1+ \frac{N_t}{w}\right) \log \frac{w}{s} \right)$. Its parallel depth for processing a single minibatch is $O(\log s)$.
\end{lemma}
%-------------------

%-------------------
\begin{proof}
For the above algorithm, consider the total work involved in processing a single bucket of $w$ elements. If the entire bucket was contained within a minibatch, the total work to process it is $\Theta(s)$ (or possibly even smaller), since no more than $s$ elements are chosen from this bucket (Observation~\ref{obs:swor-par}). Suppose that the bucket was split across two or more minibatches. The total work to maintain a sample for this bucket depends on the sizes of the minibatches that constitute this bucket -- the larger the minibatches, the closer the work is to $O(s)$. But in the worst case, each minibatch could be of size $1$, and in this case the total work is the number of times the sample changes over observing $s$ elements -- this is $s + \sum_{i=(s+1)}^w \frac{s}{i} = O(s+s\log (w/s))$. Over $t$ minibatches, the number of different buckets is $N_t/w$, which leads to a total work of $O(s N_t (\log (w/s))/w)$. We also note that $O(1)$ time is required for each minibatch. Further, in the case $N_t < w$, when a single bucket has not completed yet, we still need $O(s\log w)$ work. The parallel depth follows because for each minibatch that arrives, it is needed to determine which of Cases (1)-(3) apply here, which can be done in $O(1)$ steps, followed by updating \oldbkt~ and \newbkt. The depth of the update process for these follows from Observation~\ref{obs:swor-par} and Theorem~\ref{thm:sworinf-par}. \qed
\end{proof}
%-------------------

\section{Parallel Sampling with Replacement}
\label{sec:swr}

We now consider \swrinf, sampling with replacement. Intuitively, parallelizing
sampling with replacement is easier than without replacement because random
samples can be independently drawn without checking for distinctness. 

%------------------------
\begin{observation}
\label{obs:swr-par}
There is an $O(s)$-work $O(1)$-depth parallel algorithm for computing a sample
of size $s$ with replacement from an input array of size $n$.
\end{observation}
%------------------------

The above is easy to see: Each sample element can be independently chosen by
selecting a random integer in the interval $[1,n]$ and choosing the
corresponding element from the input. We will use the binomial random variable.
For integer $\alpha > 0$ and $0 < \beta \le 1$, let $\binomial(\alpha,\beta)$
be the the number of successes in $\alpha$ trials, where each trial has probability $\beta$ of success.\\

%We consider the problems \swrinf, \swrfixed, and \swrvar.

{\bf \swrinf: Sampling with replacement, infinite window.} A simple algorithm is
to run $s$ independent parallel copies of a single element stream sampling
algorithm, which is clearly correct. When minibatch $B_i$ is received, each
single element sampling algorithm decides whether or not to replace its sample,
with probability $n_i/N_i$, which can be done in $O(1)$ time. The algorithm has
$O(s)$ work per minibatch, and parallel depth $O(1)$. However, we can do better
than this as follows.

Suppose the current samples after observing $i$ minibatches are stored in an
array $S[1,2,\ldots,s]$. For each sample, the probability of it being replaced
after the minibatch arrives is $p_i = n_i/N_i$. Instead of testing each sample
separately to see if needs to be replaced, the algorithm carries out the
following steps:
\begin{enumerate}[label=(\Alph*)]
\item Generate a random number $s' \le s$ according to $\binomial(s,p_i)$ to
  determine how many of the samples need to be replaced.
\item Sample $L$, a set of $s' \le s$ elements \emph{without replacement} from the
  set $\{1,2,\ldots,s\}$, using a parallel sampling algorithm
  (Observation~\ref{obs:swor-par}).
\item For each location $\ell \in L$, sample $S[\ell]$ is overwritten with a
  randomly chosen element from $B_i$.
\end{enumerate}

Steps (A) and (B) in the above procedure can be shown to generate the same
distribution of locations as iterating through each location in $S$, and
separately determining whether the sample needs to be replaced. Step~(A) can be
performed in $O(1)$ work, % using the inverse CDF method,
while Step~(B) can be performed in $O(s')$ work with $O(\log s') = O(\log s)$
parallel depth. Since $\expec{s'} = sp_i = s n_i/N_i$, the expected work of this
step is $O(s n_i/N_i)$. Step~(C) can be performed in $O(s')$ work and $O(1)$
depth. Below are the overall properties of the algorithm:

%\begin{lemma}
%  \label{lem:swrinf-par}
\begin{theorem}
\label{thm:swrinf-par}
  There is a parallel algorithm for \swrinf{} such that for a target sample size
  $s$, the total work to process minibatches $B_1, \dots, B_t$ is $O(t +
  \sum_{i=1}^t s n_i/N_i)$, and the depth for processing any one minibatch $B_i$
  is $O(\log s)$. This work is optimal, up to constant factors.
%\end{lemma}
\end{theorem}
% ------------------------
% \begin{lemma}
% \label{lem:swrinf-par}
% There is a parallel algorithm for sampling $s$ elements \swrinf~that processes an incoming minibatch $B_i$ using work $O(1+ s n_i/N_i)$ and parallel depth $O(\log s)$. This work is optimal, up to constant factors.
% \end{lemma}
%------------------------

This work bound is optimal---the expected number of elements in the sample that
change due to a new minibatch is $(s n_i/N_i)$. The theorem means the
minibatches, at least the initial few, should be large. To see why: if they are
of size $1$, the first element would require $\Theta(s)$ work, because every
sample needs to be updated! Similarly, the first $s/2$ elements each requires
$O(s)$ work, totaling $O(s^2)$ work for the initial $\Theta(s)$ elements. With
minibatches sized $\Omega(s)$, the total work decreases significantly.
Importantly, this is not an artifact of our algorithm---any algorithm for
\swrinf{} needs this cost when minibatches are small.

% Thus, it is
% desired for the minibatches to be rather large; for \swrinf, at least the
% initial few minibatches must be large. To see why, consider the case when each
% minibatch is of size $1$. When the first element arrives, the work required is
% $\Theta(s)$, which is essential, since every sample needs to be updated!
% Similarly, each of the first $s/2$ elements in the stream requires $O(s)$ work
% for update, for a total work of $O(s^2)$ just for the initial $\Theta(s)$
% elements. If the minibatches are of size $\Omega(s)$, then the total work
% decreases significantly. Importantly, this is not an artifact of our
% algorithm---any algorithm for \swrinf{} needs to pay this cost when minibatches
% are small.

%{\bf  \swrinf: Sampling with replacement, fixed length sliding window.} The reduction from fixed size window to infinite window, which we used in the case of infinite window, is not efficient in this case, due to the above problem. The reduction starts a new instance of infinite window sampler every $w$ elements. Each such sampler can degenerate to have a work of $O(s^2)$, which is inefficient. Instead, we consider a different parallelization that is applicable to the more general problem of variable length sliding window sampling. 

%%% Local Variables:
%%% mode: latex
%%% TeX-master: "paper"
%%% End:

\section{Conclusion}
\label{sec:concl}

We presented low-depth, work-efficient parallel algorithms for the fundamental
data streaming problem of streaming sampling. Both the sliding-window and
infinite-window cases were addressed. Interesting directions for future work
include the parallelization of other types of streaming sampling problems, such
as weighted sampling and stratified sampling. % It would also be useful to
% understand the relation between the parallel and distributed sampling stream
% sampling models. This question is also connected to multicore processing with
% non-uniform memory access (NUMA), where the memory accesses need to be modeled
% in greater detail.

%\section*{Acknowledgment}
\let\oldbibliography\thebibliography
\renewcommand{\thebibliography}[1]{%
  \oldbibliography{#1}%
  \setlength{\itemsep}{0pt}%
}
%\setlength{\bibsep}{4pt}
%\begin{footnotesize}
\bibliographystyle{alpha}
\bibliography{sampling}
%\end{footnotesize}

\end{document}